\documentclass[envcountsame,envcountsect]{llncs}
\usepackage{ifpdf}
\ifpdf
   \usepackage{pdfsync}
\fi
\usepackage[english]{babel}
\usepackage{amsbsy,stmaryrd,amssymb,amsmath,enumerate,a4wide}

\newcommand{\NF}[3]{{\overline{#1}}\langle #2,\overline{#2}\rangle #3}
\newcommand{\nf}{\operatorname{nf}}
\newcommand{\ol}{\operatorname{ol}}
\newcommand{\bN}{\mathbb{N}}
\newcommand{\bM}{\mathbb{M}}
\newcommand{\mixed}{\operatorname{\mu}}
\newcommand{\posproj}{\operatorname{\pi}}
\newcommand{\negproj}{\operatorname{\overline{\pi}}}

\def\vec{\overrightarrow}

\begin{document}

\title{The trace monoids in the queue monoid and in the direct product
  of two free monoids} 

\author{Dietrich Kuske and Olena Prianychnykova%
  \thanks{Supported by the DFG-Project ``Speichermechanismen als
    Monoide'', KU 1107/9-1.}}
\institute{Fachgebiet Automaten und Logik,
  Technische Universit\"at Ilmenau}
\date{}
\maketitle
\begin{abstract}
  We prove that a trace monoid embeds into the queue monoid if and
  only if it embeds into the direct product of two free monoids. We
  also give a decidable characterization of these trace monoids.
\end{abstract}

\section{Introduction}

Trace monoids model the behavior of concurrent systems whose
concurrency is governed by the use of joint resources. They were
introduced into computer science by Mazurkiewicz in his study of Petri
nets~\cite{Maz77}. Since then, much work has been invested on their
structure, see \cite{DieR95} for comprehensive surveys. A basic fact
about trace monoids is that they can be embedded into the direct
product of free monoids \cite{CorP85}. Since the proof of this fact is
constructive, an upper bound for the number of factors needed in such
a free product is immediate (it is the number $\alpha$ of cliques
needed to cover the dependence alphabet). If the dependence alphabet
is a path on $n$ vertices, than this upper bound equals the exact
number, namely $n-1$. But there are cases where the exact number is
considerably smaller (the examples are from~\cite{DieMR95}:
\begin{itemize}
\item If the independence alphabet is the disjoint union of two copies
  of $C_4$ (the cycle on four vertices), then $\alpha=4$, but $3$
  factors suffice.
\item If the independence alphabet is the disjoint union of $n$ copies
  of $K_k$ (the complete graph on $k$ vertices), then $\alpha=k^n$,
  but $k$ factors suffice.
\end{itemize}
The strongest result in this respect is due to Kunc \cite{Kun04}:
Given a $C_3$- and $C_4$-free dependence alphabet and a natural number
$k$, it is decidable whether the trace monoid embeds into the direct
product of $k$ free monoids. In this paper, we extend this positive
result to all dependence alphabets, but only for the case $k=2$. More
precisely, we give a complete and decidable characterization of all
independence alphabets whose generated trace monoid embeds into the
direct product of two free monoids.

Queue monoids, another class of monoids, have been introduced
recently~\cite{HusKZ14,HusKZ16}. They model the behavior of a single
fifo-queue. Intuitively, the basic actions (i.e., generators of the
monoid) are the action of writing the letter $a$ into the queue
(denoted $a$) and reading the letter $a$ from the queue (denoted
$\bar a$). Sequences of actions are equivalent if they induce the same
state change on any queue. For instance, writing a symbol into the
queue and reading \emph{another} symbol from the other end of the
queue are two actions that can be permuted without changing the
overall behavior, symbolically: $a\bar b\equiv \bar b a$. But there
are also more complex equivalences that can be understood as
``conditional commutativity'', e.g., $a b\bar b\equiv a\bar b b$. The
unconditional commutations allow to embed the direct product of two
free monoids into the queue monoid~\cite{HusKZ16}. In \cite{HusKZ16},
it is conjectured that the monoid $\bN^3$ cannot be embedded into the
queue monoid. Note that these two monoids are special trace monoids
and that any trace monoid embedding into the direct product of two
free monoids consequently embeds into the queue monoid. In this paper,
we prove the conjecture from~\cite{HusKZ16} and characterize, more
generally, the class of trace monoids that embed into the queue
monoid.

In summary, this paper characterized two classes of trace monoids
defined by their embedability into $\{a,b\}^*\times\{c,d\}^*$ and into
the queue monoid, respectively. As it turns out, these two classes are
the same, i.e., a trace monoid embeds into the direct product of two
free monoids if and only if it embeds into the queue monoid, and this
property is decidable.

\section{Preliminaries and main result}

\subsection{The trace monoid}

Trace monoids are meant to model the behavior of concurrent systems
whose concurrency is governed by the use of joint resources. Here, we
take a slightly more abstract view and say that two actions are
independent if they use disjoint resources. More formally, an
\emph{independence alphabet} is a pair $(\Gamma,I)$ consisting of a
countable (i.e., finite or of size $\aleph_0$) set $\Gamma$ and an
irreflexive and symmetric relation $I\subseteq\Gamma^2$ called the
\emph{independence relation}. By $D=\Gamma^2\setminus I$, we denote
the complementary \emph{dependence relation}.

An independence alphabet $(\Gamma,I)$ induces a trace monoid as
follows: Let $\equiv_I$ denote the least congruence on the free monoid
$\Gamma^*$ with $ab\equiv_I ba$ for all pairs $(a,b)\in I$. Then the
\emph{trace monoid associated with $(\Gamma,I)$} is the quotient
$\bM(\Gamma,I)=\Gamma^*/\mathord{\equiv_I}$, the equivalence class
containing $u\in\Gamma^*$ is denoted $[u]_I$.

Thus the defining equations of the trace monoid are the equations
$ab\equiv_I ba$ for some pairs of letters $(a,b)$.  

We will only need very basic properties of the trace monoid
$\bM(\Gamma,I)$, namely the following:

\begin{proposition}\label{P:trace_monoid}
  Let $(\Gamma,I)$ be an independence alphabet.
  \begin{enumerate}[(1)]
  \item Let $\Gamma=\bigcup_{i\in I}C_i$ with
    $D=\bigcup_{i\in I}C_i\times C_i$. Then the trace monoid
    $\bM(\Gamma,\Gamma^2\setminus I)$ embeds into the monoid
    \[
       \prod_{i\in I}\{a,b\}^*\,,
    \]
    i.e., into a direct product of free monoids~\cite{CorP85}.
  \item The trace monoid $\bM(\Gamma,I)$ is cancellative, i.e.,
    $uvw\equiv_I uv'w$ implies $v\equiv_I v'$ for all words
    $u,v,v',w\in\Gamma^*$.
  \end{enumerate}
\end{proposition}

In this paper, we will often use graph-theoretic terms to speak about
an independence alphabet $(\Gamma,I)$ -- where we identify $I$ with
the set of edges $\{a,b\}$ for $(a,b)\in I$. In other words, we think
of $(\Gamma,I)$ as a symmetric and loop-free graph. We will also take
the liberty to write $(C,I)$ for the subgraph of $(\Gamma,I)$ induced
by $C\subseteq\Gamma$. We call a connected component $C$ of
$(\Gamma,I)$ \emph{nontrivial} if it is not an isolated vertex. The
connected component $C$ is \emph{bipartite} if
$I\cap C^2\subseteq (C_1\times C_2)\cup (C_2\times C_1)$ for some
partition $C_1\uplus C_2$ of $C$. It is \emph{complete bipartite} if
$I\cap C^2=(C_1\times C_2)\cup (C_2\times C_1)$. Finally, an
independence alphabet $(\Gamma,I)$ is \emph{$P_4$-free} if no induced
subgraph is isomorphic to $P_4$, i.e., if there are no four distinct
vertices $a,b,c,d$ with $(a,b),(b,c),(c,d)\in I$ and
$(b,d),(d,a),(a,c)\in D$.

Using this graph theoretic language, the sets $C_i$ in
Proposition~\ref{P:trace_monoid}(1) form a covering of $(\Gamma,D)$ by
cliques. It follows that the trace monoid $\bM(\Gamma,I)$ can be
embedded into the direct product of two free monoids whenever
$(\Gamma,D)$ has a clique covering with two cliques. But the existence
of a clique cover with two cliques is not necessary for such an
embedding. As an example, consider the independence alphabet
$(\Gamma,I)$ with $\Gamma=\{a_i,b_i\mid 0\le i<n\}$ and
$I=\{(a_i,b_i),(b_i,a_i)\mid 0\le i< n\}$ (where
$n\in\bN\cup\{\aleph_0\}$). Then
$D=\{(a_i,b_j),(b_j,a_i)\mid 0\le i,j< n, i\neq j\}$. Hence
$|D|+n=n^2$ and all cliques in $(\Gamma,D)$ contain at most $2$
elements. Our main result shows that, nevertheless, the trace monoid
$\bM(\Gamma,I)$ embeds into the direct product of two copies of
$\{a,b\}^*$.

\subsection{The queue monoid}
The queue monoid models the behavior of a fifo-queue whose entries
come from a set~$A$. Consequently, the state of a valid queue is an
element from~$A^*$. In order to have a defined result even if a read
action fails, we add the error state~$\bot$. The basic actions are
writing of the symbol $a\in A$ into the queue (denoted~$a$) and
reading the symbol $a\in A$ from the queue
(denoted~$\overline{a}$). Formally, $\overline{A}$ is a disjoint copy
of $A$ whose elements are denoted $\overline{a}$. Furthermore, we set
$\Sigma=A\cup\overline{A}$.  Hence, the free monoid $\Sigma^*$ is the
set of sequences of basic actions and it acts on the set
$A^*\cup \{\bot\}$ by way of the function
$.\colon (A^*\cup\{\bot\})\times\Sigma^* \to A^*\cup\{\bot\}$, which
is defined as follows:
\begin{align*}
  q.\varepsilon & = q &
  q.au&= qa.u & q.\overline{a}u &=
  \begin{cases}
    q'.u & \text{ if }q=aq'\\
    \bot & \text{ otherwise}
  \end{cases}
  & \bot.u &=\bot
\end{align*}
for $q\in A^*$, $a\in A$, and $u\in\Sigma^*$.

\begin{definition}
  Two words $u,v\in\Sigma^*$ are \emph{equivalent} if $q.u=q.v$ for
  all queues $q\in A^*$. In that case, we write $u\equiv v$. The
  equivalence class wrt.\ $\equiv$ containing the word $u$ is
  denoted~$[u]$.

  Since $\equiv$ is a congruence on the free monoid $\Sigma^*$, we can
  define the quotient monoid $Q_A=\Sigma^*/\mathord{\equiv}$ that is
  called the \emph{monoid of queue actions} or \emph{queue monoid} for
  short.
\end{definition}

Note that two queue monoids are not isomorphic if the generating sets
have different size. But, for any generating set $A$, the queue monoid
$Q_A$ embeds into $Q_{\{a,b\}}$ \cite[Cor.~5.5]{HusKZ16} (the proof in
\cite{HusKZ16} can easily be extended to infinite sets $A$). Since
this paper is concerned with submonoids of $Q_A$, the concrete size of
$A$ does not matter. Hence we will simply write $Q$ for $Q_A$, no
matter what the set $A$ is.

\begin{theorem}[\protect{\cite[Theorem 4.3]{HusKZ16}}]\label{T: equations}
  The equivalence relation $\equiv$ is the least congruence on the
  free monoid $\Sigma^*$ satisfying the following for all $a,b,c\in A$:
  \begin{align*}
    a\overline{b} &\equiv \overline{b}a \text{ if }a\neq b\\
    a\overline{b}\overline{c} & \equiv \overline{b} a\overline{c}\\
    ab\overline{c} & \equiv a\overline{c}b
  \end{align*}
\end{theorem}

The second and third of these equations generalize nicely to words:
\begin{lemma}[\protect{\cite[Corollary
    3.6]{HusKZ16}}]\label{L:generalized_equations}
  Let $u,v,w\in A^*$.
  \begin{itemize}
  \item If $|u|\le|w|$, then $u \overline{v} \overline{w} \equiv
    \overline{v} u \overline{w}$.
  \item If $|u|\ge|w|$, then $u v \overline{w} \equiv
    u \overline{w} v$.
  \end{itemize}
\end{lemma}

Let $\posproj\colon \Sigma^*\to A^*$ be the homomorphism defined by
$\posproj(a)=a$ and $\posproj(\overline{a})=\varepsilon$ for all $a\in
A$. Similarly, define the homomorphism $\negproj\colon
\Sigma^*\to A^*$ by $\posproj(a)=\varepsilon$ and $\posproj(\overline{a})=a$ for
all $a\in A$. Then, from Theorem~\ref{T: equations}, we
immediately get
\[
   u\equiv v\ \Longrightarrow\
    \posproj(u)=\posproj(v)\text{ and }\negproj(u)=\negproj(v)
\]
for all words $u,v\in\Sigma^*$. Hence the homomorphisms $\posproj$ and
$\negproj$ define homomorphisms from $Q$ to $A^*$ by
$[u]\mapsto\posproj(u)$ and $[u]\mapsto\negproj(u)$. The words
$\posproj(u)$ and $\negproj(u)$ are called the \emph{positive} and
\emph{negative projection} of $u$ (or $[u]$).

Ordering the equations from Theorem~\ref{T: equations} from left to
right, we obtain a semi-Thue system. This semi-Thue system is
confluent and terminating. Hence any equivalence class of $\equiv$ has
a unique normal form. To describe these normal forms, we write
$\left<a_1a_1\dots
  a_n,\overline{b_1}\overline{b_2}\dots\overline{b_n}\right>$ for
$a_1\overline{b_1} a_2\overline{b_2}\dots a_n\overline{b_n}$ (where
$n\in\bN$ and $a_i,b_i\in A$ for all $1\le i\le n$). Then a word
$u\in\Sigma^*$ is in normal form iff there are three words
$u_1,u_2,u_3\in A^*$ with $u= \NF{u_1}{u_2}{u_3}$. We write $\nf(u)$
for the unique word from the equivalence class $[u]$ in normal form.
Furthermore, the mixed or central part of the word $\nf(u)$, i.e., the
word $u_2$ with $\nf(u)=\NF{u_1}{u_2}{u_3}$ is denoted $\mixed(u)$.

The importance of this word $\mixed(u)$ is described by the following
observation: Let $u,v\in \Sigma^*$. Then the following are equivalent:
\begin{enumerate}
\item $u\equiv v$
\item $\nf(u)=\nf(v)$
\item $\posproj(u)=\posproj(v)$, $\negproj(u)=\negproj(v)$, and
  $\mixed(u)=\mixed(v)$
\end{enumerate}

Next, we describe the normal form of the product of two words in
normal form. For this, we need the concept of the overlap of two
words: Let $u,v\in A^*$. Then the \emph{overlap} of $u$ and $v$ is the
longest word $x$ that is both, a suffix of $u$ and a prefix of $v$. We
write $\ol(u,v)$ for this overlap.
\begin{theorem}[\protect{\cite[Theorem
    5.5]{HusKZ16}}]\label{T:multiplication}
  Let $u,v\in A^*$. Then $\nf(uv)= \NF{s}{\mixed(uv)}{t}$ with
  \begin{align*}
    \mixed(uv)&=\ol(\mixed(u)\negproj(v),\posproj(u)\mixed(v))\,,\\
    s\mixed(uv) &=\negproj(uv)\text{ and }\\
    \mixed(uv)\,t &=\posproj(uv)\,.
    \end{align*}
\end{theorem}

In the following lemma we describe the normal form of the $n$-th power
of an element of the queue monoid~$Q$. This will turn out useful in
the following considerations.

\begin{lemma}\label{L:power}
  Let $u\in A^*$. Then for every $n\geq 1$ we have 
  \[
    \mixed(u^n)=
  \ol(\mixed(u)\negproj(u)^{n-1},\posproj(u)^{n-1}\mixed(u))\,.
  \]
\end{lemma}
\begin{proof}
  We prove the lemma by induction on $n$. The statement is obvious for
  $n =1$.

  Let $n>1$ and assume that the statement holds for every $i <
  n$. Then by the induction hypothesis 
  \[
    \mixed (u^{n-1})=
    \ol(\mixed(u)\negproj(u)^{n-2},\posproj(u)^{n-2}\mixed(u))\,.
  \]
  Now set
  \[ 
    s=\ol (\mixed(u^{n-1})\negproj(u), \posproj(u^{n-1})\mixed(u))
  \]
  such that $\mixed(u^n)=\mixed(u^{n-1}u)=s$ by
  Theorem~\ref{T:multiplication}. It remains to be shown that $s$ is
  the overlap of the words $\mixed(u)\negproj(u)^{n-1}$ and
  $\posproj(u)^{n-1}\mixed(u)$. To simplify notation, let $s'$ denote
  this overlap, i.e., set
  \[
    s'=\ol(\mixed(u)\negproj(u)^{n-1},\posproj(u)^{n-1}\mixed(u))\,.
  \]
  Note that $s$ is a suffix of $\mixed(u^{n-1})\negproj(u)$. Since
  $\mixed(u^{n-1})$ is a suffix of $\mixed(u)\negproj(u)^{n-2}$, it
  follows that $s$ is a suffix of $\mixed(u)\negproj(u)^{n-1}$. By its
  very definition, $s$ is also a prefix of
  $\posproj(u^{n-1})\mixed(u)$. Since $s'$ is the longest word that is
  both, a suffix of $\mixed(u)\negproj(u)^{n-1}$ and a prefix of
  $\posproj(u)^{n-1}\mixed(u)$, it follows that $|s'|\ge|s|$. Since
  $s=\ol (\mixed(u^{n-1})\negproj(u), \posproj(u^{n-1})\mixed(u))$, we
  get $|s|-|\mixed(u^{n-1})|\le |\negproj(u)|$, i.e.,
  $|s|\le|\mixed(u^{n-1}) \negproj(u)|$. Since both, $s'$ and
  $\mixed(u^{n-1})\negproj(u)$ are suffixes of
  $\mixed(u)(\negproj(u))^{n-1}$, it follows that $s'$ is a suffix of
  $\mixed(u^{n-1}) \negproj(u)$. Since it is also a prefix of
  $\posproj(u)^{n-1}\mixed(u)$, we get $|s|\ge|s'|$. Hence we showed
  $|s|=|s'|$. Consequently, $s$ and $s'$ are prefixes of
  $\posproj(u^{n-1})\mixed(u)$ of the same length and therefore
  $s=s'$.\qed
\end{proof}

\subsection{The main result}

The results of this paper are summarised in the following theorem. It
characterizes those trace monoids that can be embedded into the queue
monoid as well as those that embed into the direct product of two free
monoids. In particular, these two classes of trace monoids are the
same. And, in addition, given a finite independence alphabet, it is
decidable whether the generated trace monoid falls into this class.

\begin{theorem}\label{T main}
  Let $(\Gamma,I)$ be a countable independence alphabet. Then the
  following are equivalent:
  \begin{enumerate}[(1)]
  \item The trace monoid $\mathbb M(\Gamma,I)$ embeds into the queue
    monoid $Q$.
  \item The trace monoid $\mathbb M(\Gamma,I)$ embeds into the direct
    product $\{a,b\}^*\times\{c,d\}^*$ of two free monoids.
  \item One of the following conditions hold:
    \begin{enumerate}[(3.a)]
    \item All nodes in $(\Gamma,I)$ have degree $\le 1$.
    \item The independence alphabet $(\Gamma,I)$ has only one
      non-trivial connected component and this component is complete
      bipartite.
    \end{enumerate}
  \end{enumerate}
\end{theorem}

The implication ``(2) implies (1)'' follows immediately from
\cite[Prop~8.2]{HusKZ14} since there, we showed that
$\{a,b\}^*\times\{c,d\}^*$ embeds into the queue monoid~$Q$. In the
following section, we present embeddings of $\bM(\Gamma,I)$ whenever
$(\Gamma,I)$ satisfies condition (3). The main work here is concerend
with independence alphabets satisfying (3.a). The subsequent section
shows that any trace monoid that embeds into the queue monoid
satisfies condition~(3). Technically, this proof is much harder than
the first one.

\section{(3) implies (2) in Theorem~\ref{T main}}

Let $(\Gamma,I)$ be an independence alphabet satisfying (3.a) or (3.b)
of Theorem~\ref{T main}. We will prove that $\mathbb M(\Gamma,I)$
embeds into the direct product of two free monoids
(Lemma~\ref{L:case_1}).

\begin{lemma}\label{L:case_1}
  Let $(\Gamma,I)$ be an (at most countably infinite) independence
  alphabet such that all nodes in $(\Gamma,I)$ have degree $\le
  1$. Then $\bM(\Gamma,I)$ embeds into the direct product of two
  countably infinite free monoids.
\end{lemma}

\begin{proof}
  Consider the independence alphabet $(\Sigma,I)$ with
  $\Sigma=\{a_i,b_i\mid i\in\bN\}$ and 
  \[
    I=\{(a_i,b_i),(b_i,a_i)\mid i\in\bN\}\,.
  \]  
  Then $(\Gamma,I)$ can be seen as a sub-alphabet of $(\Sigma,I)$ so
  that $\bM(\Gamma,I)$ embeds into $\bM(\Sigma,I)$. 

  We embed $\bM(\Sigma,I)$ into the direct product
  \[
     M=\{c_i\mid i\in\bN\}\times\{d_i\mid i\in\bN\}\,.
  \]
  Note that in this monoid $(c_i,d_i)$ and $(c_i,d_id_i)$
  commute. Hence there is a homomorphism
  $\eta\colon\bM(\Sigma,I)\to M$ with $\eta(a_i)=(c_i,d_i)$ and
  $\eta(b_i)=(c_i,d_id_i)$ for all $i\in\bN$.

  To show that this homomorphism is injective, we use lexicographic
  normal forms. So let $\sqsubseteq$ be a linear order on $\Sigma$
  with $a_i\sqsubset b_i$ for all $i\in\bN$.  Now let $u\in \Sigma^*$
  be in lexicographic normal form wrt.\ $\sqsubseteq$. Then the word
  $u$ has the form
  \[
    u= a_{i_1}^{k_1} b_{i_1}^{\ell_1} a_{i_2}^{k_2} b_{i_2}^{\ell_2}\cdots
       a_{i_s}^{k_s} b_{i_s}^{\ell_s}
  \]
  where $i_a\in\bN$, $k_a+\ell_a>0$ for all $1\le a\le s$ and
  $i_a\neq i_{a+1}$ for all $1\le a<s$. The image of $u$ equals
  \[
    \eta(u)=
    \begin{pmatrix}
      c_{i_1}^{k_1+\ell_1} & c_{i_2}^{k_2+\ell_2} & \cdots
        & c_{i_s}^{k_s+\ell_s}\\
      d_{i_1}^{k_1+2\ell_1} & d_{i_2}^{k_2+2\ell_2} & \cdots
        & d_{i_s}^{k_s+2\ell_s}\\
    \end{pmatrix}\,.
  \]
  Next let also $v$ be a word in lexicographic normal form:
  \[
    v= a_{j_1}^{m_1} b_{j_1}^{n_1} a_{j_2}^{m_2} b_{j_2}^{n_2}\cdots
       a_{j_t}^{m_t} b_{j_t}^{n_t}
  \]
  where $j_a\in\bN$, $m_a+n_a>0$ for all $1\le a\le t$ and
  $j_a\neq j_{a+1}$ for all $1\le a<t$. The image of $u'$ equals
  \[
    \eta(v)=
    \begin{pmatrix}
      c_{j_1}^{m_1+n_1} & c_{j_2}^{m_2+n_2} & \cdots & c_{j_t}^{m_t+n_t}\\ 
      d_{j_1}^{m_1+2n_1} & d_{j_2}^{m_2+2_2} & \cdots & d_{j_t}^{m_t+2_t}
    \end{pmatrix}\,.
  \]
  Suppose $\eta(u)=\eta(v)$. Since all the exponents of $c_i$ and
  $d_i$ in the expressions for $\eta(u)$ and for $\eta(v)$ are
  positive and consecutive $c_i$ and $d_i$ have distinct indices, we
  obtain $s=t$, $i_a=j_a$, $k_a+\ell_a=m_a+n_a$ and
  $k_a+2\ell_a=m_a+2n_a$ for all $1\le a\le s$. Hence $k_a=m_a$ and
  $\ell_a=n_a$ for all $1\le a\le s$ and therefore $u=v$. Hence $\eta$
  embeds $\bM(\Sigma,I)$ into $M$ and we get\\

  \hspace*{\fill} $\bM(\Gamma,I)\hookrightarrow
  \bM(\Sigma,I)\hookrightarrow M\,.$\hspace*{\fill}\qed
\end{proof}

\begin{theorem}\label{T main2}
  Let $(\Gamma,I)$ be an independence alphabet such that one of the
  following conditions holds:
  \begin{enumerate}
  \item all nodes in $(\Gamma,I)$ have degree $\le 1$ or
  \item $(\Gamma,I)$ has only one non-trivial connected component and
    this component is complete bipartite 
  \end{enumerate} 
  Then $M(\Gamma,I)$ embeds into $\{a,b\}^*\times\{c,d\}^*$.
\end{theorem}

\begin{proof}
  Let $(\Gamma,I)$ be such that the first condition holds, i.e., all
  nodes in $(\Gamma,I)$ have degree $\le 1$. Then by
  Lemma~\ref{L:case_1} there is an embedding of $M(\Gamma,I)$ into a
  direct product of two countably infinite free monoids.

  Now let $(\Gamma,I)$ be such that the second condition holds, i.e.,
  $(\Gamma,I)$ has only one non-trivial connected component and this
  component is complete bipartite. In other words,
  $\Gamma=\Gamma_1\uplus \Gamma_2\uplus \Gamma_3$ with
  $I=\Gamma_1\times\Gamma_2\cup\Gamma_2\times\Gamma_1$. Then the
  corresponding dependence alphabet $(\Gamma,D)$ can be covered by the
  two cliques induced by $\Gamma_1\cup\Gamma_3$ and
  $\Gamma_2\cup\Gamma_3$. Consequently, \cite[Corollary 1.4.5 (General
  Embedding Theorem), p.~26]{Die90} implies that $M(\Gamma,I)$ is a
  submonoid of a direct product of two countably infinite free
  monoids.

  Note that the countably infinite free monoid $\{a_i\mid i\in\bN\}^*$
  embeds into $\{a,b\}^*$ via $a_i\mapsto a^ib$. Hence, in any case,
  $\bM(\Gamma,I)$ embeds into $\{a,b\}^*\times\{c,d\}^*$.\qed
  \end{proof}

\section{(1) implies (3) in Theorem~\ref{T main}}

\begin{definition}
Let $(\Gamma,I)$ be an independence alphabet and
$\eta\colon\mathbb M(\Gamma,I)\hookrightarrow Q$ be an embedding. We
partition $\Gamma$ into sets $\Gamma_+$, $\Gamma_-$, and $\Gamma_\pm$
according to the emptiness of the projections of $\eta(a)$:
\begin{itemize}
\item $a\in\Gamma_+$ iff $\posproj(\eta(a))\neq\varepsilon$ and
  $\negproj(\eta(a))=\varepsilon$.
\item $a\in\Gamma_-$ iff $\posproj(\eta(a))=\varepsilon$ and
  $\negproj(\eta(a))\neq\varepsilon$.
\item $a\in\Gamma_\pm$ iff $\posproj(\eta(a))\neq\varepsilon$ and
  $\negproj(\eta(a))\neq\varepsilon$.
\end{itemize}
\end{definition}

We will prove the following:
\begin{itemize}
\item $(\Gamma_+\cup\Gamma_-,I)$ is complete bipartite
  (Proposition~\ref{P: Gamma+ cup Gamma-}).
\item Every node $a\in\Gamma_\pm$ has degree $\le1$ (Corollary~\ref{C:
    small degree} which is the most difficult part of the proof).
\item Any letter from $\Gamma_+\cup\Gamma_-$ is connected to any edge
  (Proposition~\ref{P: P2 cup P3}).
\item The graph $(\Gamma,I)$ is $P_4$-free (Proposition~\ref{P:
    P4-free}).
\end{itemize}
At the end of this section, we infer that the independence alphabet
$(\Gamma,I)$ has the required property from Theorem~\ref{T main}~(3).

\subsection{The set $\Gamma_+\cup\Gamma_-$ induces a complete
  bipartite subgraph of $(\Gamma,I)$}

\begin{proposition}\label{P: Gamma+ cup Gamma-}
  Let $(\Gamma,I)$ be an independence alphabet, let
  $\eta\colon\mathbb{M}(\Gamma,I)\hookrightarrow Q$ be an embedding . 

  Then $(\Gamma_+,I)$ and $(\Gamma_-,I)$ are discrete and
  $(\Gamma_+\cup\Gamma_-,I)$ is complete bipartite.
\end{proposition}

\begin{proof}
  We first show that $(\Gamma_+,I)$ is discrete.

  Towards a contradiction, suppose there are $a,b\in\Gamma_+$ with
  $(a,b)\in I$. Let $u=\posproj(\eta(a))$ and $v=\posproj(\eta(b))$. Since
  $\posproj\circ\eta\colon\mathbb M(\Gamma,I)\to A^*$ is a homomorphism and
  since $[ab]_I=[ba]_I$, we get $uv=vu$. Hence $u$ and $v$ have a
  common root, i.e., there is a word $p$ and there are $i,j>0$ with
  $u=p^i$ and $v=p^j$. Hence
  \[
     \posproj(\eta(a)^j) = u^j = v^i = \posproj(\eta(b)^i)\,.
  \]
  Clearly, we also have
  \[
     \negproj(\eta(a)^j) = \varepsilon = \negproj(\eta(b)^i)\,.
  \]
  Hence 
    \[
     \eta(a)^j= [u^j]= [v^i] =\eta(b)^i.
  \]
 Since $\eta$ is
  injective, this implies $a^i\equiv_I b^j$ and therefore $a=b$,
  contradicting $(a,b)\in I$. Hence, there are no $a,b\in\Gamma_+$
  with $(a,b)\in I$, i.e., $(\Gamma_+,I)$ is discrete.

  Symmetrically, also $(\Gamma_-,I)$ is discrete.

  It remains to be shown that $(a,b)\in I$ for any $a\in\Gamma_+$ and
  $b\in\Gamma_-$. So let $a\in\Gamma_+$ and $b\in\Gamma_-$. Then there
  are words $u,v\in A^*$ with $\eta(a)=[u]$ and
  $\eta(b)=[\overline{v}]$ (note that $u$ and $v$ are nonempty since
  $\eta$ is an injection). We have the following:
  \begin{align*}
    \eta(abb^{|u|}) &= [u\overline{v} \overline{v}^{|u|}]\\
      &= [\overline{v}u \overline{v}^{|u|}] 
        &\text{by Lemma~\ref{L:generalized_equations} since }|u|\le|v^{|u|}|\\
      &= \eta(bab^{|u|})
  \end{align*}
  Since $\eta$ is injective, this implies $abb^{|u|}\equiv_I
  bab^{|u|}$ and therefore $ab\equiv_I ba$. Now $(a,b)\in I$ follows
  from $a\neq b$.\qed
\end{proof}

\subsection{Nodes from $\Gamma_+\cup\Gamma_-$ are connected to any edge}

\begin{lemma}\label{L:nontrivial_equation_P2_cup_P3}
  Let $u,v,w\in \Sigma^+$ such that $\negproj(u)=\varepsilon$,
  $vw\equiv wv$ and $v\neq w$. Then there exist vectors
  $\vec{x}=(x_u,x_v,x_w)$ and $\vec{y}=(y_u,y_v,y_w)$ in $\bN^3$ such
  that $x_v+x_w\neq0$ and
  \begin{equation}
    \label{eq: nontrivial equation P_2 cup P_3}
    u^{x_u} v^{x_v} u w^{x_w} \equiv u^{y_u}  w^{y_w} u v^{y_v}\,.
  \end{equation}
  (Note that the two sides of this equation differ in particular in
  the order of the words $v$ and $w$.)
\end{lemma}

\begin{proof}
  Since $vw\equiv wv$, there exist primitive words $p$ and $q$ and
  natural numbers $a_v,a_w,b_v,b_w$ satisfying the following:
  \begin{align*}
    \posproj(v) &= p^{a_v} & \posproj(w) &= p^{a_w}\\
    \negproj(v) &= q^{b_v} & \negproj(w) &= q^{b_w}
  \end{align*}
  Since $v,w\neq\varepsilon$, we get $a_v+b_v\neq 0 \neq a_w+b_w$.

  We first show that there are natural numbers $x_v,x_w,y_v,y_w$ (not
  all zero) that satisfy the following system of linear equations.
  \begin{equation}
    \label{eq:linear-system1}
    \left.
    \begin{array}{rcl}
    a_v x_v &=& a_w y_w\\
    a_w x_w &=& a_v y_v\\
    b_v x_v + b_w x_w &=& b_w y_w + b_v y_v
    \end{array}
    \right\}
  \end{equation}
  If $a_v=0$, then set $x_v=y_v=1$ and $x_w=y_w=0$. Symmetrically, if
  $a_w=0$, we set $x_v=y_v=0$ and $x_w=y_w=1$. If $a_v b_w = a_w b_v$,
  then set $x_v=y_v=a_w+b_w>0$ and $x_w=y_w=a_v+b_v>0$.

  Now consider the case $a_v\neq 0\neq a_w$ and $a_vb_w\neq
  a_wb_v$. The system~\eqref{eq:linear-system1} has a nontrivial
  solution over the field $\mathbb Q$. Consequently, there are
  integers $x_v,x_w,y_v,y_w$ (not all zero) satisfying these
  equations. We show $x_v>0\iff x_w>0$: First note that $x_v\neq 0$
  iff $y_w\neq0$ and $x_w\neq0$ iff $y_v\neq 0$. Since not all of the
  integers $x_v,x_w,y_v,y_w$ are zero, we get $x_v\neq0$ or
  $x_w\neq0$. Furthermore, since we have a solution, we get
  \[
    y_w = \frac{a_v}{a_w}x_v\text{ and }
    y_v = \frac{a_w}{a_v}x_w\,.
  \]
  Substituting these into the third equation yields
  \[
    (b_v-b_w\frac{a_v}{a_w})\cdot x_v = (b_v\frac{a_w}{a_v}-b_w)\cdot x_w
    =(b_v-b_w\frac{a_v}{a_w})\cdot\frac{a_w}{a_v}\cdot x_w\,.
  \]
  From $a_vb_w\neq a_wb_v$, we get
  $b_v-b_w\frac{a_v}{a_w}\neq0$. Hence $x_v=\frac{a_w}{a_v}\cdot x_w$
  and therefore $a_vx_v=a_wx_w$ follow. Now $a_v,a_w>0$ imply
  $x_v>0\iff x_w>0$. Consequently, all of $x_v,x_w,y_v,y_w$ are
  non-negative or all are non-positive. Hence $|x_v|,|x_w|,|y_v|,|y_w|$
  is a solution to the system~\eqref{eq:linear-system1} in natural
  numbers as required.

  From now on, let $x_v,x_w,y_v,y_w\in\bN$ be a nontrivial solution of
  the system~\eqref{eq:linear-system1}. Furthermore, let
  $x_u=y_u\in\bN$ such that
  $|\negproj(v^{x_v}uw^{x_w})|\le|u|\cdot x_u=|u^{x_u}|$. Then we have
  the following:
  \begin{align*}
    u^{x_u} v^{x_v} u w^{x_w}
     &\equiv u^{x_u}\,\overline{\negproj(v^{x_v} u w^{x_w})}
                    \, \posproj(v^{x_v} u w^{x_w})
        &\text{by Lemma~\ref{L:generalized_equations}}\\
     &= u^{x_u}\,\overline{q}^{b_v x_v + b_w x_w} 
               \,p^{a_v x_v} u p^{a_w x_w}\\
     &= u^{y_u}\,\overline{q}^{b_w y_w + b_v y_v}
               \,p^{a_w y_w} u p^{a_v y_v}\\
     &= u^{y_u}\,\overline{\negproj(w^{y_w} u v^{y_v})}
               \,\posproj(w^{y_w} u v^{y_v})\\
     &\equiv u^{y_u} w^{y_w} u v^{y_v}
        &\text{by Lemma~\ref{L:generalized_equations}}
  \end{align*}
  Thus, we found the vectors $\vec{x}$ and $\vec{y}$ satisfying
  Equation~\eqref{eq: nontrivial equation P_2 cup P_3} with
  $x_v+x_w\neq0$.\qed
\end{proof}

\begin{proposition}\label{P: P2 cup P3}
  Let $(\Gamma,I)$ be an independence alphabet and let
  $\eta\colon\mathbb{M}(\Gamma,I)\hookrightarrow Q$ be an
  embedding. Let $a\in\Gamma_+\cup\Gamma_-$ and $b,c\in \Gamma$ with
  $(b,c)\in I$. Then $(a,b)\in I$ or $(a,c)\in I$.
\end{proposition}

\begin{proof}
  If $a\in\{b,c\}$, we get $(a,b)\in I$ or $(a,c)\in I$ from $(b,c)\in
  I$. So assume $a\notin\{b,c\}$. There are words $u,v,w\in\Sigma^+$
  with $\eta(a)=[u]$, $\eta(b)=[v]$, and $\eta(c)=[w]$. Since
  $(b,c)\in I$, we get $[vw]=\eta(bc)=\eta(cb)=[wv]$ and therefore
  $vw\equiv wv$. Furthermore, $[v]=\eta(b)\neq\eta(c)=[w]$ since
  $\eta$ is injective and since $b\neq c$ follows from $(b,c)\in
  I$. Hence in particular $v\neq w$.  

  We first consider the case $a\in\Gamma_+$, i.e.,
  $\negproj(u)=\varepsilon$.  From
  Lemma~\ref{L:nontrivial_equation_P2_cup_P3}, we find natural numbers
  $x_u,x_v,x_w,y_u,y_v,y_w$ with
  $u^{x_u} v^{x_v} u w^{x_w} \equiv u^{y_u} w^{y_w} u v^{y_v}$ and
  $x_v+x_w+y_v+y_w\neq 0$. Consequently,
  \begin{align*}
    \eta(a^{x_u} b^{x_v} a c^{x_w})
      &= [u^{x_u} v^{x_v} u w^{x_w} ]\\
      &= [u^{y_u} w^{y_w} u v^{y_v}]\\
      &=\eta(a^{y_u} c^{y_w} a b^{y_v})\,.
  \end{align*}
  Since $\eta$ is injective, this implies
  \[
    a^{x_u} b^{x_v} a c^{x_w} \equiv_I a^{y_u} c^{y_w} a b^{y_v}\,.
  \]
  If $x_v\neq0$, then $(a,b)\in I$. Similarly, if $x_w\neq0$, then
  $(a,c)\in I$. This settles the case $\negproj(u)=\varepsilon$.

  Now let $\posproj(u)=\varepsilon$. By duality,
  Lemma~\ref{L:nontrivial_equation_P2_cup_P3} yields natural numbers
  $x_u,x_v,x_w,y_u,y_v,y_w$ with $x_v+x_w+y_v+y_w\neq0$ and
  $v^{x_v} u w^{x_w} u^{x_u} \equiv w^{y_w} u v^{y_v} u^{y_u}$. Then
  we can derive $(a,b)\in I$ or $(a,c)\in I$ as above.\qed
\end{proof}


\subsection{Nodes from $\Gamma_\pm$ have degree $\le1$}

Let $a\in\Gamma_\pm$. Then there are nonempty primitive words $p$ and
$q$ with $\posproj(\eta(a))\in p^+$ and $\negproj(\eta(a))\in q^+$,
i.e., $p$ and $q$ are the primitive roots of the two projections of
$\eta(a)$. The proof of the fact that $a$ has at most one neighbor in
$(\Gamma,I)$ distinguishes two cases: first, we handle the case that
$p$ and $q$ are not conjugated (recall that $p$ and $q$ are
\emph{conjugated} iff there are words $g\in A^*$ and $h\in A^+$ with
$p=gh$ and $q=hg$). The second case, namely that $p$ and $q$ are
conjugated, turns out to be far more difficult.

\subsubsection{Non-conjugated roots}

\begin{proposition}\label{P: nonconjugated roots}
  Let $(\Gamma,I)$ be an independence alphabet and let
  $\eta\colon\mathbb{M}(\Gamma,I)\hookrightarrow Q$ be an
  embedding. Let furthermore $b\in\Gamma$ and $p,q\in A^+$ be
  primitive with $p\not\sim q$ such that
  \[
     \posproj(\eta(b))\in p^+\text{ and }\negproj(\eta(b))\in q^+\,.
  \]
  Then there is at most one letter $a\in\Gamma$ with $(a,b)\in I$.
\end{proposition}

\begin{proof}
 Towards a contradiction, suppose there are distinct letters $a$ and
  $c$ in $\Gamma$ with $(a,b),(b,c)\in I$. Let
  \[
     u=\nf(\eta([ab]_I))\,,\ v=\nf(\eta(b))\,,
    \text{ and }w=\nf(\eta([bc]_I))\,.
  \]
  Since $(a,b)\in I$, we have $ab\equiv_I ba$ and therefore
  $\eta([ab]_I)=\eta([ba]_I)$. This implies
  $\posproj(\eta(a))\,\posproj(\eta(b))=\posproj(\eta(b))\,\posproj(\eta(a))$,
  i.e., the two words $\posproj(\eta(a))$ and $\posproj(\eta(b))$
  commute in the free monoid. Since $\posproj(\beta(b))\in p^+$ and
  $p$ is primitive, this implies $\posproj(\eta(a))\in p^*$ and
  therefore $\posproj(u)=\posproj(\eta(a))\,\posproj(\eta(b))\in
  p^+$. Similarly, $\negproj(u)\in q^+$ as well as
  $\posproj(w)\in p^+$ and $\negproj(w)\in q^+$. Hence there are
  positive natural numbers $a_u,a_v,a_w,b_u,b_v,b_w$ such that the
  following hold:
  \begin{align*}
    \posproj(u)&= p^{a_u} & \posproj(v)&= p^{a_v} &  \posproj(w)&= p^{a_w} \\
    \negproj(u)&= q^{b_u} & \negproj(v)&= q^{b_v}
      & \negproj(w)&= q^{b_w}
  \end{align*}
 
First we prove that there exist vectors  $\vec{x}=(x_u,x_v,x_w)\in \bN^3$ and $\vec{y}=(y_u,y_v,y_w)\in
  \bN^3$ with $\vec{x}\neq\vec{y}$ such that
 \begin{equation}
    \label{eq: nontrivial equation - nonconjugated case}
    u^{x_u}v^{x_v} w^{x_w} \equiv u^{y_u}v^{y_v} w^{y_w}\,.
  \end{equation}

  Consider the following system of linear equations:
  \begin{equation}
    \label{eq:linear-system - nonconjugated case}
    \left.
    \begin{array}{rcl}
      a_ux_u + a_vx_v + a_wx_w & = & a_uy_u + a_vy_v + a_wy_w\\
      b_ux_u + b_vx_v + b_wx_w & = & b_uy_u + b_vy_v + b_wy_w
    \end{array}
    \right\}
  \end{equation}
  Using Gaussian elimination, we find a nontrivial rational
  solution. Hence the system \eqref{eq:linear-system - nonconjugated
    case} has an integer solution. Increasing all entries in the
  integer solution by some fixed number $n\in\bN$ yields another
  solution. Hence we can choose $n$ large enough such that the
  resulting solution $\vec{x}=(x_u,x_v,x_w)$ and
  $\vec{y}=(y_u,y_v,y_w)$ satisfies
  \begin{itemize}
  \item $\vec{x},\vec{y}\in\bN^3$
  \item $|p|+|q|\le b_w\cdot x_w\cdot |q|$ and
    $|p|+|q|\le b_w\cdot y_w\cdot |q|$, and
  \item $|p|+|q|\le (a_u\cdot x_u+a_v\cdot x_v)\cdot|p|$ and
    $|p|+|q|\le (a_u\cdot y_u+a_v\cdot y_v)\cdot|p|$.
  \end{itemize}

  Now we show that $\vec{x}$ and $\vec{y}$ is a solution to the
  Equation~\eqref{eq: nontrivial equation - nonconjugated case}.

  First, we have
  \begin{align*}
    \posproj(u^{x_u} v^{x_v} w^{x_w})
      &= (p^{a_u})^{x_u} (p^{a_v})^{x_v} (p^{a_w})^{x_w}\\
      &= p^{a_u x_u + a_v x_v + a_w x_w}\\
      &= p^{a_u y_u + a_v y_v + a_w y_w}\\
      &= \posproj(u^{y_u} v^{y_v} w^{y_w})
   \intertext{and similarly}
    \negproj(u^{x_u} v^{x_v} w^{x_w})
      &= \negproj(u^{y_u} v^{y_v} w^{y_w}) \,.
  \end{align*}
 
  It remains to be shown that $\mixed(u^{x_u} v^{x_v} w^{x_w})$ equals
  $\mixed(u^{y_u} v^{y_v} w^{y_w})$. Let $H$ denote the set of words
  that are both, a suffix of $q^m$ and a prefix of $p^n$ for some
  $m,n\in\bN$. First note that $\mixed(u)$ belongs to $H$ since it is
  a suffix of $\negproj(u)=q^{b_u}$ and a prefix of
  $\posproj(u)=p^{a_u}$. By Lemma~\ref{L:power}, 
  \[
    \mixed(u^{x_u}=\ol(\mixed(u) q^{b_u(x_u-1)}, p^{a_u(x_u-1)} \mixed(u))
  \]
  is a suffix of $\mixed(u) q^{b_u(x_u-1)}$ which is a suffix of $q^m$
  for some $m\in\bN$ since $u\in H$. Symmetrically, $\mixed(u^{x_u})$
  is a prefix of $p^{a_u(x_u-1)}\mixed(u)$ and therefore a prefix of
  $p^m$ for some $m\in\bN$ since $\mixed(u)\in H$. Hence we get
  $\mixed(u^{x_u})\in H$. Using the analogous arguments, it follows
  that
  \[
     \mixed(u^{x_u} v^{x_v})
       =\ol(\mixed(u^{x_u})\,q^{b_vx_v},   p^{a_ux_u}\,\mixed(v^{x_v}))\,.
  \]
  belongs to $H$. Finally, also
  \[
    \mixed(u^{x_u} v^{x_v} w^{x_w})
      = \ol(\mixed(u^{x_u} v^{x_v})q^{b_wx_w},p^{a_ux_u+a_vx_v} \mixed(w^{x_w}))
  \]
  is an element of $H$ by analogous arguments. For our following
  argument, it is important to note that
  $\mixed(u^{x_u} v^{x_v} w^{x_w})$ is a factor of $q^m$ and of $p^m$
  for some $m\in\bN$. Since $p\not\sim q$, \cite[Lemma7,
  p.282]{LohM11} implies
  $|\mixed(u^{x_u} v^{x_v} w^{x_w})|\le|p|+|q|$.  Furthermore, we have
  $|p|+|q|\le b_w\cdot x_w\cdot|q|=|q^{b_w x_w}|$ and
  $|p|+|q|\le (a_ux_u+a_vx_v)\cdot|p|=|p^{a_ux_u+a_vx_v}|$ and
  therefore $|\mixed(u^{x_u} v^{x_v} w^{x_w})|\le |q^{b_w x_w}|$ and
  $|\mixed(u^{x_u} v^{x_v} w^{x_w})|\le |p^{a_ux_u+a_vx_v}|$.
  Consequently,
  \begin{align*}
    \mixed(u^{x_u} v^{x_v} w^{x_w})
      &=\ol(\mixed(u^{x_u} v^{x_v})q^{b_wx_w},
            p^{a_ux_u+a_vx_v} \mixed(w^{x_w}))\\
      &=\ol(q^{b_wx_w},p^{a_ux_u+a_vx_v})\\
      &=\ol(q^{b_ux_u+b_vx_v+b_wx_w},p^{a_ux_u+a_vx_v+a_wx_w})\\
      &=\ol(q^{b_uy_u+b_vy_v+b_wy_w},p^{a_uy_u+a_vy_v+a_wy_w})\,.
  \end{align*}
  By symmetric arguments, this last overlap equals
  $\mixed(u^{y_u} v^{y_v} w^{y_w})$. Thus, indeed,
  \[
    \mixed(u^{x_u} v^{x_v} w^{x_w})=\mixed(u^{y_u} v^{y_v} w^{y_w})\,.
  \]

  Hence the two words $u^{x_u} v^{x_v} w^{x_w}$ and
  $u^{y_u} v^{y_v} w^{y_w}$ agree in their projections and their
  normal forms agree in their mixed part. Consequently, the normal
  forms of these two words coincide. Hence they are equivalent, i.e.,
  as required, we found a non-trivial solution $\vec{x}$, $\vec{y}$ of
  Equation~\eqref{eq: nontrivial equation - nonconjugated case}.

  Finally we obtain
  \begin{align*}
     \eta([(ab)^{x_u} b^{x_v} (bc)^{x_w}]_I)
     &=[u^{x_u} v^{x_v} w^{x_w}]\\
     &=[u^{y_u} v^{y_v} w^{y_w}]\\
     &=\eta([(ab)^{y_u} b^{y_v} (bc)^{y_w}]_I)\,.
  \end{align*}
  Since $\eta$ is injective, and since $(a,b),(b,c)\in I$, this
  implies
  \begin{align*}
    a^{x_u} b^{x_u+x_v+x_w} c^{x_w}
     &\equiv_I (ab)^{x_u} b^{x_v} (bc)^{x_w}\\
     &\equiv_I (ab)^{y_u} b^{y_v} (bc)^{y_w}\\
     &\equiv_I a^{y_u} b^{y_u+y_v+y_w} c^{y_w}\,.
  \end{align*}
  Since the letters $a$, $b$, and $c$ are mutually distinct, we obtain
  \[
    (x_u,x_u+x_v+x_w,x_w)= (y_u,y_u+y_v+y_w,y_w)
  \]
  and therefore $\vec{x}=\vec{y}$. But this contradicts our choice of
  these two vectors as distinct.  Thus there are no two distinct
  letters $a$ and $c$ with $(a,b),(b,c)\in I$.\qed
\end{proof}

Note that the above proof, essentially, proceeded as follows: we aimed
at a nontrivial solution to Equation \eqref{eq: nontrivial equation -
  nonconjugated case} in natural numbers. Length conditions on the
positive and negative projections yielded the system of linear
equations \eqref{eq:linear-system - nonconjugated case}. Since this
system consists of two equations in the unknown $x_u-y_u$, $x_v-y_v$
and $x_w-y_w$, it has an integer solution that can be increased by
arbitrary natural numbers, i.e., there is a ``sufficiently large''
solution that makes the positive (and negative) projections of
$u^{x_u} v^{x_v} w^{x_w}$ and $u^{y_u} v^{y_v} w^{y_w}$ equal. Using
that this solution is ``sufficiently large'' and that $p$ and $q$ are
not conjugated, we employed some combinatorics on words to prove that
also the mixed parts of the normal forms of these two words were
equal.

\subsubsection{Conjugated roots}

We now want to prove a similar result in case $p$ and $q$ are
conjugated. The proof, although technically more involved, will
proceed similarly, i.e., we will determine a non-trivial solution of
Equation~\eqref{eq: nontrivial equation - nonconjugated case}.  But
presentationwise, we will proceed differently: First,
Lemma~\ref{L:product_of_powers} describes the mixed part of the normal
form of $u^{x_u} v^{x_v} w^{x_w}$. Then,
Lemma~\ref{L:nontrivial_equation} determines a nontrival solution to
(some rotation of) Equation~\eqref{eq: nontrivial equation -
  nonconjugated case}, before, finally, Proposition~\ref{P: conjugated
  roots} proves the analogous to Proposition~\ref{P: nonconjugated
  roots} for conjugated roots.

We first prove a combinatorial lemma on words that are prefix of some
power of $p$ and, at the same time, suffixes of some power of $q$
(where $p$ and $q$ are conjugated).

\begin{lemma}\label{L:conj}
  Let $g \in A^*$, $h \in A^+$ such that $p=gh$ and $q=hg$ are both
  primitive words. Let furthermore $y$ be some suffix of $q^i$ and
  some prefix of $p^j$ for some $i,j \geq 1$ such that $|y| \geq
  |q|$. Then $y=g q^k=p^k g$ where
  $k=\left\lfloor\frac{|y|}{|q|}\right\rfloor$.
\end{lemma}

\begin{proof}
  Since $y$ is a suffix of $q^i$, there exist words $r\in A^+$ and
  $s\in A^*$ with $y=sq^k$ and $q=rs$.

  Since $p$ and $q$ are conjugate, their lengths are equal. Hence
  $k=\left\lfloor\frac{|y|}{|p|}\right\rfloor$. Since $y$ is a prefix
  of $p^j$, there exist words $s'\in A^*$ and $t\in A^+$ with
  $y=p^k s'$ and $p=s't$.

  Since $|p|=|q|$, $sq^k=y=p^k s'$ implies $|s|=|s'|$. Together with
  $s(rs)^k=sq^k=y = p^k s' = (s't)^k s'=s'(ts')^k$, this implies
  $s=s'$. Since $k>0$, we also get $r=t$. Hence we obtained $q=rs$ and
  $p=s't=sr$. Since $p$ and $q$ are conjugate primitive words and
  $r\in A^+$, \cite[Proposition 1.3.3, p.~8]{Lot83} implies
  $(g,h)=(s,r)$. This ensures in particular $g=s$ and therefore
  $y=g q^k=p^k g$.\qed
\end{proof}

Using this combinatorial lemma, we can often determine the overlap of
two words via the following corollary:

\begin{corollary}\label{C: fact}
  Let $g \in A^*$, $h \in A^+$ such that $p=gh$ and $q=hg$ are both
  primitive words. Furthermore, let $p'$ be a suffix of $p$ with
  $|p'|<|p|$ and let $q'$ be a prefix of $q$ with $|q'|<|q|$.

  Then for every $i,j\in \bN$ we have $\ol(p'g q^i,p^j g q')=g
  q^{min(i,j)}$.
\end{corollary}
\begin{proof}
  Let $y=\ol(p' g q^i,p^jgq')$. Since $p'$ is a suffix of $p=gh$, the
  word $p' g q^i$ is a suffix of $g h g q^i= gq^{i+1}$ and therefore
  of $q^{i+2}$. Hence also $y$ is a suffix of $q^{i+2}$. Similarly,
  $y$ is a prefix of $p^{j+2}$. By Lemma~\ref{L:conj}, we obtain $y=g
  q^k=p^k g$ for some $k\in\bN$ and it remains to be shown that
  $k=\min(i,j)$.

  Note that
  \begin{align*}
    k|q|+|g| &= |y|\\
             &\le |p' g q^i| & \text{ since $y$ is a suffix of }p' g q^i\\
             &< (i+1)|q|+|g| & \text{ since }|p'|<|p|=|q|\,.
  \end{align*}
  This implies $k\le i$ and, similarly, we can show $k\le j$, i.e.,
  $k\le \min(i,j)$. On the other hand note that $g
  q^{\min(i,j)}=p^{\min(i,j)} g$ is a suffix of $p' g q^i$ and a
  prefix of $p^j g q'$ implying $k\ge\min(i,j)$ since $gq^k=\ol(p' g
  q^i,p^j g q')$. Hence $k=\min(i,j)$.\qed
\end{proof}

\begin{lemma}\label{L:product_of_powers}
  Let $g\in A^*$, $h\in A^+$ such that $p=gh$ and $q=hg$ are
  primitive. Let $u,v,w\in Q$ such that the following holds for some
  $a_u,a_v,a_w,b_u,b_v,b_w\in\bN\setminus\{0\}$, and
  $c_u,c_v,c_w\in\mathbb{Z}$:
  \begin{align*}
    \posproj(u)&= p^{a_u} & \negproj(u)&= q^{b_u} &
         c_u&= \begin{cases}
                -1 & \text{ if }|\mixed(u)|<|g|\\
                \left\lfloor{\frac {|\mixed(u)|} {|q|} }\right\rfloor
                   &\text{ otherwise}
               \end{cases}\\
    \posproj(v)&= p^{a_v} & \negproj(v)&= q^{b_v} &
         c_v&= \begin{cases}
                -1 & \text{ if }|\mixed(v)|<|g|\\
                \left\lfloor{\frac {|\mixed(v)|} {|q|} }\right\rfloor
                   &\text{ otherwise}
               \end{cases}\\
    \posproj(w)&= p^{a_w} & \negproj(w)&= q^{b_w} &
         c_w&= \begin{cases}
                -1 & \text{ if }|\mixed(w)|<|g|\\
                \left\lfloor{\frac {|\mixed(w)|} {|q|} }\right\rfloor
                   &\text{ otherwise}
               \end{cases}
  \end{align*}

  Let $\vec{x}=(x_u,x_v,x_w)\in \bN^3$ with $x_u,x_v,x_w\ge2$. Then
  $\mixed(u^{x_u} v^{x_v} w^{x_w})=gq^{X_{\vec{x}}}=p^{X_{\vec{x}}}g$
  where
  \[
  X_{\vec{x}}=\min
  \left( \begin{array}{rrrr}
    \min(a_u, b_u) x_u+ & b_v x_v+ & b_w  x_w + & c_u-\min(a_u, b_u),\\
    a_u  x_u+ & \min(a_v, b_v) x_v+ & b_w  x_w + & c_v-\min(a_v, b_v), \\
    a_u  x_u+ & a_v  x_v + & \min(a_w, b_w) x_w + & c_w-\min(a_w, b_w)
  \end{array}
   \right)\,.
  \]
\end{lemma}

\begin{proof}
  From Lemma~\ref{L:power}, we get
  \begin{align*}
    \mixed(u^{x_u})&= \ol(c(u)\negproj(u)^{x_u-1},\posproj(u)^{x_u-1}\mixed(u))\,.
  \end{align*}
  Depending on the length of $\mixed(u)$, we distinguish three cases:
  First, let $|\mixed(u)|<|g|$. Since $\mixed(u)$ is a suffix of
  $\negproj(u)\in q^*=(hg)^*$, the word $\mixed(u)$ is a suffix of
  $g$. Similarly, $\mixed(u)$ is a prefix of $\posproj(u)\in
  p^*=(gh)^*$ implying that $\mixed(u)$ is a prefix of $g$. Then
  $a_u,b_u>0$ and $x_u\ge2$ imply $b_u(x_u-1),a_u(x_u-1)>0$. Hence we
  can determine $\mixed(u)$ as follows:
  \begin{align*}
    \mixed(u^{x_u}) &=\ol(\mixed(u) q^{b_u(x_u-1)}, p^{a_u(x_u-1)} \mixed(u))\\
      &=\ol(\mixed(u) hg q^{b_u(x_u-1)-1}, p^{a_u(x_u-1)-1}gh \mixed(u))\\
      &=gq^{\min(b_u(x_u-1)-1,a_u(x_u-1)-1)}
        &\text{by Corollary~\ref{C: fact}}\\
      &=gq^{\min(a_u,b_u)\cdot(x_u-1)+c_u}
        &\text{since $c_u=-1$}
  \end{align*}
  Next, consider the case $|g|\le|\mixed(u)|<|q|$. Then $\mixed(u)$ is
  a prefix of $p=gh$ and a suffix of $q=hg$. Hence there are a prefix
  $h'$ and a suffix $h''$ of $h$ with $\mixed(u)=gh'=h''g$. Now we can
  determine $\mixed(u^{x_u}$ as follows:
  \begin{align*}
    \mixed(u^{x_u}) &=\ol(\mixed(u) q^{b_u(x_u-1)}, p^{a_u(x_u-1)}\mixed(u))\\
      &=\ol(h'g q^{b_u(x_u-1)}, p^{a_u(x_u-1)}gh'')\\
      &=gq^{\min(b_u(x_u-1),a_u(x_u-1))}
        &\text{by Corollary~\ref{C: fact}}\\
      &=gq^{\min(a_u,b_u)\cdot(x_u-1)+c_u}
        &\text{since $c_u=0$}
  \end{align*}
  Finally, let $|q|\le|\mixed(u)|$. Then
  $c_2=\left\lfloor\frac{|\mixed(u)|}{|q|}\right\rfloor$. Furthermore,
  $\mixed(u)$ is a prefix of $\posproj(u)\in p^*$ and a suffix of
  $\negproj(u)\in q^*$. Hence, by Lemma~\ref{L:conj},
  $\mixed(u)=gq^{c_u}=p^{c_u}g$. Hence we can determine $\mixed(u^{x_u})$ as
  follows:
  \begin{align*}
    \mixed(u^{x_u}) &=\ol(\mixed(u) q^{b_u(x_u-1)}, p^{a_u(x_u-1)}\mixed(u))\\
      &=\ol(g q^{c_u+b_u(x_u-1)}, p^{a_u(x_u-1)+c_u}g)\\
      &=gq^{\min(c_u+b_u(x_u-1),a_u(x_u-1)+c_u)}
        &\text{by Corollary~\ref{C: fact}}\\
      &=gq^{\min(a_u,b_u)\cdot(x_u-1)+c_u}
  \end{align*}
  In other words, we proved
  \[
    \mixed(u^{x_u})=g q^{e_u}=p^{e_u} g
  \]
  with
  \[
    e_u=\min(a_u,b_u)\cdot(x_u-1)+c_u\,.
  \]
  Clearly, similar statements hold for $\mixed(v^{x_v})$ and
  $\mixed(w^{x_w})$.\bigskip

  In a second step, we determine $\mixed(u^{x_u} v^{x_v})$. We
  get
  \begin{align*}
    \mixed(u^{x_u} v^{x_v})
     &=\ol(\mixed(u^{x_u})\,\negproj(v^{x_v}),\posproj(u^{x_u})\,\mixed(v^{x_v}))\\
      &=\ol(gq^{e_u} q^{b_vx_v},
                p^{a_ux_u} p^{e_v}g)\\
      &=gq^{\min(e_u+b_vx_v,a_ux_u+e_v)}\,.
  \end{align*}
  In other words,
  \[
    \mixed(u^{x_u} v^{x_v})= g q^{e_{uv}} = p^{e_{uv}}g
  \]
  with
  \[
    e_{uv}= \min(e_u+b_vx_v,a_ux_u+e_v) \,.
  \]\bigskip

  In a third and last step, we determine $\mixed(u^{x_u} v^{x_v} w^{x_w})$.
  Note that $\mixed(w^{x_w})=p^{e_w}g$. Then we get
  \begin{align*}
    \mixed(u^{x_u} v^{x_v} w^{x_w})&=
      \ol(\mixed(u^{x_u} v^{x_v}) \negproj(w^{x_w}),
                   \posproj(u^{x_u} v^{x_v}) \mixed(w^{x_w}))\\
      &=\ol(g q^{e_{uv}}q^{b_wx_w},p^{a_ux_u+a_vx_v}\,q^{e_w}g)\\
      &=gq^{\min(e_{uv}+b_wx_w,a_ux_u+a_vx_v+e_w)}\,.
  \end{align*}

  Unraveling the definitions of $e_u$, $e_v$, $e_w$, and $e_{uv}$
  yields
  \begin{align*}
    \min(e_{uv}+b_wx_w,a_ux_u+a_vx_v+e_w)
      &=\min\left( \begin{array}{l}
                     \min(e_u+b_vx_v,a_ux_u+e_v)+b_wx_w,\\
                     a_u  x_u+  a_v  x_v + e_w
                   \end{array}
            \right)\\
      &=\min\left( \begin{array}{l}
                     e_u+b_vx_v+b_wx_w,\\
                     a_ux_u+e_v+b_wx_w,\\
                     a_u  x_u+  a_v  x_v + e_w
                   \end{array}
            \right)\\
      &=X_{\vec{x}}\,.
  \end{align*}

  Hence, we have indeed
  $\mixed(u^{x_u} v^{x_v} w^{x_w})=gq^{X_{\vec{x}}}$.\qed
\end{proof}

\begin{lemma}\label{L:nontrivial_equation}
  Let $g\in A^*$, $h\in A^+$ such that $p=gh$ and $q=hg$ are
  primitive. Let $u',v',w'\in \Sigma^+$ with
  $\posproj(u'),\posproj(v'),\posproj(w')\in p^+$ and
  $\negproj(u'),\negproj(v'),\negproj(w')\in q^+$.

  Then there exist a rotation $(u,v,w)$ of $(u',v',w')$ and vectors
  $\vec{x}=(x_u,x_v,x_w)\in \bN^3$ and $\vec{y}=(y_u,y_v,y_w)\in
  \bN^3$ with $\vec{x}\neq\vec{y}$ such that
  \begin{equation}
    \label{eq: nontrivial equation}
    u^{x_u}v^{x_v} w^{x_w} \equiv u^{y_u}v^{y_v} w^{y_w}\,.
  \end{equation}
\end{lemma}

\begin{proof}
  We choose the rotation $(u,v,w)$ such that one of the following three
  conditions hold:
  \begin{enumerate}
  \item $|\posproj(u)|=|\negproj(u)|$,
    $|\posproj(v)|=|\negproj(v)|$, and $|\posproj(w)|=|\negproj(w)|$
    or
  \item $|\posproj(u)|>|\negproj(u)|$ or
  \item $|\posproj(w)|<|\negproj(w)|$.
  \end{enumerate}
  Given this rotation, we define the natural numbers
  $a_u,a_v,a_w,b_u,b_v,b_w,c_u,c_v,c_w$ as in Lemma~\ref{L:product_of_powers}.

  Consider the following system of linear equations:
  \begin{equation}
    \label{eq:linear-system}
    \left.
    \begin{array}{rcl}
      a_ux_u + a_vx_v + a_wx_w & = & a_uy_u + a_vy_v + a_wy_w\\
      b_ux_u + b_vx_v + b_wx_w & = & b_uy_u + b_vy_v + b_wy_w
    \end{array}
    \right\}
  \end{equation}
  Using Gaussian elimination, we find a nontrivial rational
  solution. Hence the system \eqref{eq:linear-system} has an integer
  solution. Increasing all entries in this solution by the minimal
  entry plus 2 yields a nontrivial solution
  $\vec{x'}=(x'_u,x'_v,x'_w)$ and $\vec{y'}=(y'_u,y'_v,y'_w)$ with
  $\vec{x'},\vec{y'}\in\bN^3$ and $x'_u,x'_v,x'_w,y'_u,y'_v,y'_w\ge2$.

  From this solution of the system \eqref{eq:linear-system} of linear
  equations, we now construct a nontrivial solution $\vec{x}$,
  $\vec{y}$ that, in addition, satisfies
  $X_{\vec{x}}=X_{\vec{y}}$. This is done by considering the three
  possible cases for the rotation $(u,v,w)$ separately.

  First, let $|\posproj(u)|=|\negproj(u)|$,
  $|\posproj(v)|=|\negproj(v)|$, and $|\posproj(w)|=|\negproj(w)|$,
  i.e., $a_u=b_u$, $a_v=b_v$, and $a_w=b_w$. We obtain for the values
  $X_{\vec{x'}}$ and $X_{\vec{y'}}$ from Lemma~\ref{L:product_of_powers}:
  \begin{align*}
    X_{\vec{x'}}
     &=\min
  \left( \begin{array}{rrrr}
    a_u x'_u+ & a_v x'_v+ & a_w  x'_w + & c_u-a_u\\
    a_u  x'_u+ & a_v x'_v+ & a_w  x'_w + & c_v-a_v\\
    a_u  x'_u+ & a_v  x'_v + & a_w x'_w + & c_w-a_w
  \end{array}
   \right)\\
     &=\min
  \left( \begin{array}{rrrr}
    a_u y'_u+ & a_v y'_v+ & a_w  y'_w + & c_u-a_u\\
    a_u  y'_u+ & a_v y'_v+ & a_w  y'_w + & c_v-a_v\\
    a_u  y'_u+ & a_v  y'_v + & a_w y'_w + & c_w-a_w
  \end{array}
   \right)\\
    &= X_{\vec{y'}}\,.
  \end{align*}
  This solves the first case.

  Now, suppose $|\posproj(u)|>|\negproj(u)|$ and therefore $a_u>b_u$.
  Then we find $k\ge0$ such that the following hold:
  \begin{align*}
    b_u(x'_u+k) + b_v x'_v + b_w x'_w +c_u-\min(a_u,b_u)
     &\le a_u(x'_u+k) + \min(a_v,b_v) x'_v + b_w x'_w +c_v-\min(a_v,b_v)\\
    b_u(x'_u+k) + b_v x'_v + b_w x'_w +c_u-\min(a_u,b_u)
     &\le a_u(x'_u+k) + a_v x'_v + \min(a_w,b_w) x'_w +c_w-\min(a_w,b_w)\\
    b_u(y'_u+k) + b_v y'_v + b_w y'_w +c_u-\min(a_u,b_u)
     &\le a_u(y'_u+k) + \min(a_v,b_v) y'_v + b_w y'_w +c_v-\min(a_v,b_v)\\
    b_u(y'_u+k) + b_v y'_v + b_w y'_w +c_u-\min(a_u,b_u)
     &\le a_u(y'_u+k) + a_v y'_v + \min(a_w,b_w) y'_w +c_w-\min(a_w,b_w)
  \end{align*}
  The reason is that in all cases, when increasing $k$, the right-hand
  side grows faster than the left-hand side. Set
  \[
    \vec{x}=(x'_u+k,x'_v,x'_w)\text{ and }
    \vec{y}=(y'_u+k,y'_v,y'_w)\,.
  \]
  Then this pair of vectors forms a non-trivial solution of the
  system~\eqref{eq:linear-system}. Since $b_u=\min(a_u,b_u)$, as a
  consequence we get in addition
  \begin{align*}
    X_{\vec{x}} &=
    b_u x_u + b_v x_v + b_w x_w +c_u-\min(a_u,b_u)\\
    &=b_u y_u + b_v y_v + b_w y_w +c_u-\min(a_u,b_u)\\
    &= X_{\vec{y}}\,.
  \end{align*}
  This solves the second case.

  Finally, suppose $|\posproj(w)|<|\negproj(w)|$ and therefore
  $a_w<b_w$.  The argument now is dual to the previous case: We
  find $k\ge0$ such that the following hold:
  \begin{align*}
    a_u x'_u + a_v x'_v + a_w (x'_w + k) +c_u-\min(a_w,b_w)
     &\le \min(a_u,b_u) x'_u + b_v x'_v + b_w (x'_w+k) +c_u-\min(a_u,b_u)\\
    a_u x'_u + a_v x'_v + a_w (x'_w + k) +c_u-\min(a_w,b_w)
     &\le a_u x'_u + \min(a_v,b_v) x'_v + b_w (x'_w+k) +c_v-\min(a_v,b_v)\\
    a_u y'_u + a_v y'_v + a_w (y'_w + k) +c_u-\min(a_w,b_w)
     &\le \min(a_u,b_u) y'_u + b_v y'_v + b_w (y'_w+k) +c_u-\min(a_u,b_u)\\
    a_u y'_u + a_v y'_v + a_w (y'_w + k) +c_u-\min(a_w,b_w)
     &\le a_u y'_u + \min(a_v,b_v) y'_v + b_w (y'_w+k) +c_v-\min(a_v,b_v)
  \end{align*}
  The reason is that in all cases, when increasing $k$, the right-hand
  side grows faster than the left-hand side. This time, set
  \[
    \vec{x}=(x'_u,x'_v,x'_w+k)\text{ and }
    \vec{y}=(y'_u,y'_v,y'_w+k)\,.
  \]
  Then this pair of vectors forms a non-trivial solution of the
  system~\eqref{eq:linear-system}. Since $a_w=\min(a_w,b_w)$, as a
  consequence we get in addition
  \begin{align*}
    X_{\vec{x}} &=
    a_u x_u + a_v x_v + a_w x_w +c_w-\min(a_w,b_w)\\
    &=    a_u y_u + a_v y_v + a_w y_w +c_w-\min(a_w,b_w)\\
    &= X_{\vec{y}}\,.
  \end{align*}
  This solves the third and last case.

  So far, we constructed a nontrivial solution $\vec{x}$, $\vec{y}$
  with natural coefficients of the system~\eqref{eq:linear-system}
  that, in addition, satisfies $X_{\vec{x}}=X_{\vec{y}}$. Furthermore,
  all entries in these two vectors are at least $2$. We finally show
  that this is a solution to the Equation~\eqref{eq: nontrivial equation}:

  First, we have
  \begin{align*}
    \posproj(u^{x_u} v^{x_v} w^{x_w})
      &= (p^{a_u})^{x_u} (p^{a_v})^{x_v} (p^{a_w})^{x_w}\\
      &= p^{a_u x_u + a_v x_v + a_w x_w}\\
      &= p^{a_u y_u + a_v y_v + a_w y_w}\\
      &= \posproj(u^{y_u} v^{y_v} w^{y_w})
   \intertext{and similarly}
    \negproj(u^{x_u} v^{x_v} w^{x_w})
      &= \negproj(u^{y_u} v^{y_v} w^{y_w}) \,.
  \end{align*}
  By Lemma~\ref{L:product_of_powers}, we get
  \begin{align*}
    \mixed(u^{x_u} v^{x_v} w^{x_w})
      &= gq^{X_{\vec{x}}}\\
      &= gq^{X_{\vec{y}}}\\
      &= \mixed(u^{y_u} v^{y_v} w^{y_w})\,.
  \end{align*}
  Hence the two words $u^{x_u} v^{x_v} w^{x_w}$ and $u^{y_u} v^{y_v}
  w^{y_w}$ agree in their projections and their normal forms agree in
  their mixed part. Consequently, the normal forms of these two words
  coincide. Hence they are equivalent, i.e., as required, we found a
  non-trivial solution $\vec{x}$, $\vec{y}$ of equation
  Equation~\eqref{eq: nontrivial equation}.\qed
\end{proof}

\begin{proposition}\label{P: conjugated roots}
  Let $(\Gamma,I)$ be an independence alphabet and let
  $\eta\colon\mathbb{M}(\Gamma,I)\hookrightarrow Q$ be an
  embedding. Let furthermore $b\in\Gamma$ and $p,q\in A^+$ be
  primitive with $p\sim q$ such that
  \[
     \posproj(\eta(b))\in p^+\text{ and }\negproj(\eta(b))\in q^+\,.
  \]
  Then there is at most one letter $a\in\Gamma$ with $(a,b)\in I$.
\end{proposition}

\begin{proof}
  Towards a contradiction, suppose there are distinct letters $a$ and
  $c$ in $\Gamma$ with $(a,b),(b,c)\in I$. Let
  \[
     u'=\nf(\eta([ab]_I))\,,\ v'=\nf(\eta(b))\,,
    \text{ and }w'=\nf(\eta([bc]_I))\,.
  \]
  Since $(a,b)\in I$, we have $ab\equiv_I ba$ and therefore
  $\eta([ab]_I)=\eta([ba]_I)$. This implies
  $\posproj(\eta(a))\,\posproj(\eta(b))=\posproj(\eta(b))\,\posproj(\eta(a))$, i.e., the
  two words $\posproj(\eta(a))$ and $\posproj(\eta(b))$ commute in the free
  monoid. Since $\posproj(\beta(b))\in p^+$ and $p$ is primitive, this
  implies $\posproj(\eta(a))\in p^*$ and therefore
  $\posproj(u')=\posproj(\eta(a))\,\posproj(\eta(b))\in p^+$. Similarly,
  $\negproj(u')\in q^+$ as well as $\posproj(w')\in p^+$ and
  $\negproj(w')\in q^+$.

  Hence, by Lemma~\ref{L:nontrivial_equation}, there exists a
  rotation $(u,v,w)$ of $(u',v',w')$ and distinct vectors $\vec{x}$,
  $\vec{y}\in\bN^3$ satisfying Equation~\eqref{eq: nontrivial
    equation}. We consider the three possible rotations separately.

  First suppose the rotation is trivial, i.e., $(u,v,w)=(u',v',w')$.
  Then we obtain
  \begin{align*}
     \eta([(ab)^{x_u} b^{x_v} (bc)^{x_w}]_I)
     &=[u^{x_u} v^{x_v} w^{x_w}]\\
     &=[u^{y_u} v^{y_v} w^{y_w}]\\
     &=\eta([(ab)^{y_u} b^{y_v} (bc)^{y_w}]_I)\,.
  \end{align*}
  Since $\eta$ is injective, and since $(a,b),(b,c)\in I$, this
  implies
  \begin{align*}
    a^{x_u} b^{x_u+x_v+x_w} c^{x_w}
     &\equiv_I (ab)^{x_u} b^{x_v} (bc)^{x_w}\\
     &\equiv_I (ab)^{y_u} b^{y_v} (bc)^{y_w}\\
     &\equiv_I a^{y_u} b^{y_u+y_v+y_w} c^{y_w}\,.
  \end{align*}
  Since the letters $a$, $b$, and $c$ are mutually distinct, we obtain
  \[
    (x_u,x_u+x_v+x_w,x_w)= (y_u,y_u+y_v+y_w,y_w)
  \]
  and therefore $\vec{x}=\vec{y}$. But this contradicts our choice of
  these two vectors as distinct.

  Secondly, suppose $(u,v,w)=(v',w',u')$.
  Then we obtain
  \begin{align*}
     \eta([b^{x_u}(bc)^{x_v} (ab)^{x_w}]_I)
     &=[u^{x_u} v^{x_v} w^{x_w}]\\
     &=[u^{y_u} v^{y_v} w^{y_w}]\\
     &=\eta([b^{y_u}(bc)^{y_v} (ab)^{y_w}]_I)\,.
  \end{align*}
  As in the previous case, injectivity of $\eta$ and commutation of $b$
  with $a$ and with $c$ yields
  \[
    c^{x_v} b^{x_u+x_v+x_w} a^{x_w} \equiv_I
     c^{y_v} b^{y_u+y_v+y_w} a^{y_w}\,.
  \]
  From the distinctness of $a$, $b$ and $c$, we again get
  $\vec{x}=\vec{y}$ which contradicts our choice of these two vectors
  as distinct.

  Finally, suppose $(u,v,w)=(w',u',v')$.  Then we obtain
  \begin{align*}
     \eta([(bc)^{x_u} (ab)^{x_v} b^{x_w}]_I)
     &=[u^{x_u} v^{x_v} w^{x_w}]\\
     &=[u^{y_u} v^{y_v} w^{y_w}]\\
     &=\eta([(bc)^{y_u} (ab)^{y_v} b^{y_w}]_I)\,.
  \end{align*}
  As in the previous cases, this yields a contradiction to our choice
  of the two vectors $\vec{x}$ and $\vec{y}$ as distinct.

  Thus, indeed, there are no two distinct letters $a$ and $c$ with
  $(a,b),(b,c)\in I$.\qed
\end{proof}

The following corollary is the main result of this section. Its proof
is an immediate consequence of Propositions~\ref{P: nonconjugated
  roots} and \ref{P: conjugated roots} (depending on whether the roots
of the two projections of $\eta(a)$ are conjugated or not).

\begin{corollary}\label{C: small degree}
  Let $(\Gamma,I)$ be an independence alphabet, let $\eta\colon\mathbb
  M(\Gamma,I)\hookrightarrow Q$ be an embedding, and let $a\in\Gamma$.
  If $\posproj(\eta(a))\neq\varepsilon$ and
  $\negproj(\eta(b))\neq\varepsilon$, then the degree of $a$ is
  $\le1$.
\end{corollary}

\subsection{$(\Gamma,I)$ is $P_4$-free}

\begin{lemma}\label{L: P4-free}
  Let $t,u,v,w\in \Sigma^+$ such that $\negproj(u)=\varepsilon$,
  $\posproj(v)=\varepsilon$, $vw\equiv wv$, and $tu\equiv ut$. Then
  there exists a tuple $\vec{x}=(x_t,x_{u_1},x_{u_2},x_v,x_w)$ of
  natural numbers with $x_t, x_w\neq0$ and
  \begin{equation}
    \label{eq: second equation}
    u^{x_{u_1}} v^{x_v} w t^{x_t}  w^{x_w} u^{x_{u_2}} \equiv 
    u^{x_{u_1}} w u^{x_{u_2}}  w^{x_w} t^{x_t} v^{x_v}\,.
  \end{equation}
\end{lemma}

\begin{proof}
  Since $\negproj(u)=\varepsilon$ and $\posproj(v)=\varepsilon$, there
  are primitive words $p$ and $q$ and natural numbers $a_u,b_v>0$ with
  \[
     u=\posproj(u)=p^{a_u}\text{ and }v=\negproj(v)=q^{b_v}\,.
  \]
  Since $tu\equiv ut$ and $vw\equiv wv$, there are $a_t,b_w\in\bN$ with 
  \[
    \posproj(t)=p^{a_t}\text{ and }\negproj(w)=q^{b_w}\,.
  \]
  Then we have
  \begin{align*}
    \posproj(v^{b_w} w t^{a_u} w^{b_v} u^{a_t})
      &= \varepsilon^{b_w} \posproj(w) p^{a_ta_u} \posproj(w^{b_v}) p^{a_ua_t}\\
      &= \posproj(w) p^{a_ua_t} \posproj(w^{b_v}) p^{a_ta_u} \varepsilon\\
      &= \posproj(w u^{a_t} w^{b_v} t^{a_u} v^{b_w})
   \intertext{and}
    \negproj(v^{b_w} w t^{a_u} w^{b_v} u^{a_t})
      &= q^{b_vb_w} q^{b_w} \negproj(t^{a_u}) q^{b_wb_v} \varepsilon\\
      &= q^{b_w} \varepsilon q^{b_wb_v} \negproj(t^{a_u}) q^{b_vb_w}\\
      &= \negproj(w u^{a_t} w^{b_v} t^{a_u} v^{b_w})\,.
  \end{align*}
  Let $y\in\bN$ such that
  $|\negproj(v^{b_w} w t^{a_u} w^{b_v} u^{a_t})|=|\negproj(w u^{a_t}
  w^{b_v} t^{a_u} v^{b_w})|\le |u^y|$. We obtain
  \begin{align*}
    u^y v^{b_w} w t^{a_u} w^{b_v} u^{a_t}
      &\equiv u^y\, \overline{\negproj(v^{b_w} w t^{a_u} w^{b_v} u^{a_t})}
                 \, \posproj(v^{b_w} w t^{a_u} w^{b_v} u^{a_t})
        &&\text{by Lemma~\ref{L:generalized_equations}}\\
      &= u^y\,\overline{\negproj(w u^{a_t} w^{b_v} t^{a_u} v^{b_w})}
            \,\posproj(w u^{a_t} w^{b_v} t^{a_u} v^{b_w})\\
      &\equiv u^y w u^{a_t} w^{b_v} t^{a_u} v^{b_w}
        &&\text{by Lemma~\ref{L:generalized_equations}}\,.
  \end{align*}
  Hence the tuple $(x_t,x_{u_1},x_{u_2},x_v,x_w)=(a_u,y,a_t,b_w,b_v)$
  has the desired properties.\qed
\end{proof}

\begin{proposition}\label{P: P4-free}
  Let $(\Gamma,I)$ be an independence alphabet and let
  $\eta\colon\mathbb{M}(\Gamma,I)\hookrightarrow Q$ be an embedding.
  Then $(\Gamma,I)$ is $P_4$-free.
\end{proposition}

\begin{proof}
  Suppose there are mutually distinct nodes $a,b,c,d\in\Gamma$ with
  $(a,b),(b,c),(c,d)\in I$. Then $b$ and $c$ both have degree $\ge2$
  in $(\Gamma,I)$, i.e., they belong to $\Gamma_+\cup\Gamma_-$ by
  Corollary~\ref{C: small degree}. Since $(\Gamma_+,I)$ and
  $(\Gamma_-,I)$ are both discrete by Proposition~\ref{P: Gamma+ cup
    Gamma-}, we can assume w.l.o.g.\ that $b\in\Gamma_+$ and
  $c\in\Gamma_-$.

  There are words $t,u,v,w\in\Sigma^+$ with $\eta(a)=[t]$,
  $\eta(b)=[u]$, $\eta(c)=[v]$, and $\eta(d)=[w]$.

  Since $(a,b)\in I$, we get $[tu]=\eta(ab)=\eta(ba)=[ut]$ and therefore
  $tu\equiv ut$.
  Since $(c,d)\in I$, we get $[vw]=\eta(cd)=\eta(dc)=[wv]$ and therefore
  $vw\equiv wv$.

  Since $b\in\Gamma_+$, we get
  $\negproj(u)=\negproj(\eta(b))=\varepsilon$. Similarly, from
  $c\in\Gamma_-$, we obtain
  $\posproj(v)=\posproj(\eta(c))=\varepsilon$.

  From Lemma~\ref{L: P4-free}, we find natural numbers
  $x_t,x_{u_1},x_{u_2},x_v,x_w$ such that $x_t, x_w\neq0$ and
  \[
    u^{x_{u_1}} v^{x_v}w t^{x_t} w^{x_w} u^{x_{u_2}} \equiv
    u^{x_{u_1}} w u^{x_{u_2}} w^{x_w} t^{x_t} v^{x_v}\,.
  \]

  Consequently,
  \begin{align*}
    \eta(b^{x_{u_1}} c^{x_v}d a^{x_t} d^{x_w} b^{x_{u_2}})
      &=[u^{x_{u_1}} v^{x_v}w t^{x_t} w^{x_w} u^{x_{u_2}}  ]\\
      &=    [u^{x_{u_1}} w u^{x_{u_2}} w^{x_w} t^{x_t} v^{x_v}]\\
      &=\eta(b^{x_{u_1}} d b^{x_{u_2}} d^{x_w} a^{x_t} c^{x_v})\,.
  \end{align*}
  Since $\eta$ is injective, this implies
  \[
     b^{x_{u_1}} c^{x_v}d a^{x_t} d^{x_w} b^{x_{u_2}} 
     \equiv_I b^{x_{u_1}} d b^{x_{u_2}} d^{x_w} a^{x_t} c^{x_v}\,.
  \]
  Since $x_t,x_w\neq0$ and $a\neq d$, we obtain $(a,d)\in I$. Hence
  the mutually disjoint nodes $a,b,c,d$ do not induce $P_4$ in
  $(\Gamma,I)$.\qed
\end{proof}

\subsection{Proof of the implication (1)$\Rightarrow$(3) in
  Theorem~\ref{T main}}

\begin{theorem}\label{T main1}
  Let $(\Gamma,I)$ be an independence alphabet and $\eta\colon\mathbb
  M(\Gamma,I)\to Q$ be an embedding. Then one of the following
  conditions holds:
  \begin{enumerate}
  \item all nodes in $(\Gamma,I)$ have degree $\le 1$ or
  \item $(\Gamma,I)$ has only one non-trivial connected component and
    this component  is complete bipartite.
  \end{enumerate}
\end{theorem}

\begin{proof}
  Suppose $(\Gamma,I)$ contains a node $a$ of degree $\ge2$. Then, by
  Corollary~\ref{C: small degree}, $a\in\Gamma_+\cup\Gamma_-$. From
  Proposition~\ref{P: P2 cup P3}, we obtain that $a$ is connected to
  any edge, i.e., it belongs to the only nontrivial connected
  component $C$ of $(\Gamma,I)$. Note that $|C|\ge3$ since it contains
  $a$ and its $\ge2$ neighbors. Hence the induced subgraph $(C,I)$
  contains at least one edge. Therefore Proposition~\ref{P: P2 cup P3}
  implies $\Gamma_+\cup\Gamma_-\subseteq C$. Note that all nodes in
  $C\setminus(\Gamma_+\cup\Gamma_-)$ have degree 1 by
  Corollary~\ref{C: small degree}. Hence, by Proposition~\ref{P:
    Gamma+ cup Gamma-}, the connected graph $(C,I)$ is a complete
  bipartite graph together with some additional nodes of degree 1. It
  follows that $(C,I)$ is bipartite. By Proposition~\ref{P: P4-free},
  it is a connected and $P_4$-free graph. Hence its complementary
  graph $(C,D)$ is not connected \cite{Sei74}. But this implies that
  $(C,I)$ is complete bipartite.\qed
\end{proof}

\end{document}